\documentclass[pdflatex,sn-nature]{sn-jnl}
\usepackage{graphicx}
\usepackage{amsmath,amssymb,amsfonts}
\usepackage{amsthm}
\usepackage{booktabs}
\theoremstyle{thmstyleone}
\newtheorem{theorem}{Theorem}
\newtheorem{proposition}[theorem]{Proposition}
\theoremstyle{thmstyletwo}
\newtheorem{remark}[theorem]{Remark}
\newcommand{\Sv}{\,\mathrm{Sv}}
\newcommand{\FovS}{F_{\mathrm{ovS}}}
\newcommand{\FH}{F_{H}}
\newcommand{\R}{\mathbb{R}}

\begin{document}

\title[Memory operator ensembles and AMOC criticality]{Memory operator
ensembles indicate proximity to criticality in simulated AMOC transitions}

\author*[1]{\fnm{Mauricio} \sur{Herrera-Mar\'in}}\email{mherrera@udd.cl}
\affil*[1]{\orgdiv{Faculty of Engineering}, \orgname{Universidad del
Desarrollo}, \orgaddress{\city{Santiago}, \country{Chile}}}

\abstract{Classical early-warning signals assume that recovery slows near a tipping
point, but in a well-studied CESM simulation of AMOC collapse, autocorrelation
and variance decrease instead. We develop a preregistered diagnostic that uses
only past data to estimate the slowest mode of a memory operator fitted to the
evolving overturning structure. The operator ensemble produces an exploratory
alarm 77 years before collapse, compared with 29 years for a physics-based
freshwater-transport indicator. A blind test on the recovery branch does not
validate the same timing rule, supporting interpretation of the diagnostic as
a proximity measure rather than a universal clock. In nine additional
simulations, it warns before all four finite-rate forced transitions but not
before three rapid noise-induced collapses available only as scalar AMOC
records. Matched ablations show that the signal depends on long-timescale
memory and vertical overturning structure rather than trend, forcing, or
multivariate information alone.}


\maketitle

\section{Introduction}
\label{sec:intro}

The Atlantic Meridional Overturning Circulation is a major regulator of North
Atlantic climate, and a transition to a strongly weakened state would
reorganize ocean heat transport and regional climate. Yet no observed AMOC
transition exists against which an early-warning method can be objectively
tested. Paleoclimatic evidence together with hierarchy-spanning
modeling supports the existence of a weak or collapsed circulation state
\cite{stommel1961,rahmstorf1996,vanwesten2023}, although whether the real system
can reach a \emph{fully} collapsed state is itself contested: a recent 34-model
analysis of extreme-forcing experiments finds Southern Ocean wind-driven upwelling
sustaining a weakened but nonzero overturning in every case \cite{baker2025}. Whether an approach to an AMOC
transition could be detected in advance, and from what kind of data, has become
one of the sharpest questions in climate dynamics
\cite{lenton2011,boers2021,ditlevsen2023,benyami2024}. Any proposed detection
method faces an uncomfortable evidential situation: the real ocean provides no
ground-truthed transitions, observational records are short relative to the
relevant time scales, and uncertainties may preclude reliable extrapolation
\cite{benyami2024}. Simulated transitions with known outcomes are therefore
the only setting in which detection methods can be tested rather than merely
proposed, and the multi-millennial CESM quasi-equilibrium hosing experiment of
\cite{vanwesten2023,vanwesten2024}---a state-of-the-art global model with a
dated collapse and a dated recovery---provides one of the most stringent
available laboratories for testing early-warning methods.

Existing early-warning approaches fall into three strands. The first is the
\emph{statistical} strand built on critical slowing down: as a system
approaches a fold, its dominant restoring rate vanishes, and lag-one
autocorrelation and variance of a scalar observable are expected to rise
\cite{held2004,scheffer2009,dakos2012,kuehn2011}. These indicators are
attractive because they are model-free, and they underpin the prominent
observational analyses of \cite{boers2021} and \cite{ditlevsen2023}. Their
limitations are equally well documented: they presuppose that the scalar
observable behaves like a one-dimensional system near a normal form, they
compress all restoring processes into a single coefficient, and their
statistical power on realistic records is contested \cite{benyami2024}.
In coupled-model hosing experiments their detectability is mixed, depending
on the variable, location, model, and forcing protocol \cite{boulton2014}.
Decisively for the present setting, they \emph{fail} on the CESM collapse:
\cite{vanwesten2024} report that variance and autocorrelation of the AMOC
strength decrease toward the tipping point, a finding we replicate with a strictly one-sided pipeline below.

The second strand is \emph{physics-based}: indicators derived from the
salt-advection feedback, most prominently the freshwater transport by the
overturning at the Atlantic's southern boundary, $\FovS$, whose sign and
minimum diagnose the stability regime
\cite{rahmstorf1996,devries2005,vanwesten2024,vanwesten2025}. On the CESM collapse the $\FovS$ minimum precedes tipping by roughly a quarter
century and is a leading physics-based benchmark for this experiment. Its strengths and limits mirror the statistical
strand's in reverse: it encodes real mechanism, but it requires knowing the
mechanism, it is specific to one feedback, and---as we quantify below---its
real-time use entails a nontrivial false-turning-point cost that retrospective
analyses do not expose.

The third strand, from which this work descends, is the theory of
\emph{projected dynamics with memory}. The Mori--Zwanzig formalism shows
exactly how eliminating unresolved variables converts their influence on a
resolved observable into a memory term plus fluctuations
\cite{mori1965,zwanzig1973}, and data-driven memory closures have repeatedly
improved reduced models of climate and geophysical dynamics
\cite{kondrashov2015,falkena2019,lin2021}. Two further recent lines bear directly on what follows: data-driven approaches
now extend well beyond scalar statistics, including rare-event sampling of
internally generated AMOC transitions \cite{cini2024}; and it has become clear
that in high-dimensional systems not every observable carries critical slowing
down, tying detectability to the choice of resolved coordinates
\cite{lohmann2025,smolders2025}. What the memory strand has so far not
supplied is an early-warning \emph{diagnostic}: an object, computable from
past data alone, whose value has stability semantics for the underlying
transition.

This is the gap addressed here. For a resolved observable closed with a stable
finite-dimensional memory lift, local stability is an exact spectral property
of the lifted operator, and the balance that vanishes at a stationary
transition decomposes exactly into instantaneous drift plus integrated memory
feedback (Section~\ref{sec:theory}). The slow multiplier of a memory operator fitted using past data only is therefore a candidate indicator with three properties the
statistical strand lacks: it has an exact spectral meaning within the fitted lifted model (distance of a
spectrum to the unit circle, tested against true stability in controlled systems below), it separates delayed feedback from instantaneous
damping, and it does not presuppose scalar normal-form behavior of the
observable.

A method evaluated on a single simulated transition invites silent tuning and
selective reporting. We therefore adopt, to our knowledge for the first time in
this literature, a full preregistration discipline: alarm rules, thresholds,
calibration procedures, and decision gates frozen in writing before
evaluation; every change a numbered, dated amendment committed publicly before
the corresponding computation; verdicts reported exactly as produced,
including failures; and blind out-of-sample tests whose components were
frozen and publicly timestamped before any value of the target data was read
(Supplementary Note~2).

Here we ask whether the stability of a data-fitted memory operator can provide
early warning when scalar critical-slowing-down indicators fail. We combine an
exact spectral formulation with strictly one-sided estimation, preregistered
evaluation, and blind testing. Across the main CESM hysteresis experiment and
nine additional simulations, the results identify both a finite-rate regime in
which the indicator provides usable warning and failure modes associated with
rapid noise-induced transitions and limited observation.

\section{Results}
\label{sec:results}

\begin{figure}[t]
\centering
\includegraphics[width=.98\textwidth]{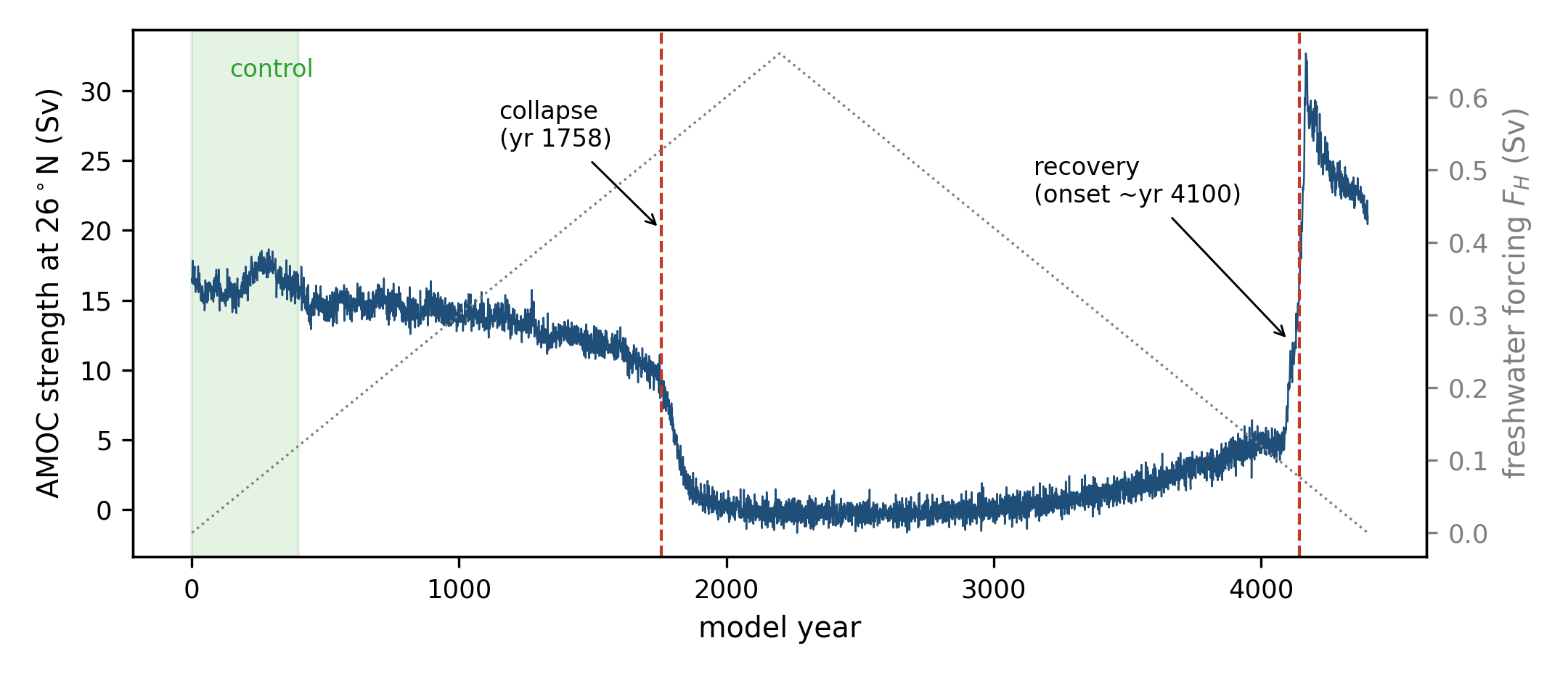}
\caption{The CESM quasi-equilibrium hysteresis experiment. AMOC strength at
$26^\circ$N (blue, left axis) and freshwater forcing $\FH$ (dotted gray, right
axis) over the full 4{,}400 model years. The control segment used for all
threshold calibration is shaded green; the collapse (year 1758) and the
recovery are indicated.}
\label{fig:overview}
\end{figure}

Figure~\ref{fig:overview} summarizes the 4{,}400-year CESM hysteresis
experiment underlying all tests below; data provenance, parity verification,
and processing are described in Methods.

\subsection{Exact stability criteria for projected dynamics with memory}
\label{sec:theory}

This subsection makes the paper self-contained; a companion manuscript
develops the continuous-time counterpart, sampling consistency, and
controlled verifications in thermohaline models \cite{herrera2026}.

Let $q_n\in\R^{d}$ be the resolved annual observable and let the unresolved
influence be represented by a linear memory lift $z_n\in\R^{h}$,
\begin{equation}
q_{n+1}=A\,q_n+B\,z_n+F u_n+b_0+\epsilon_{n+1},
\qquad
z_{n+1}=C\,q_n+D\,z_n,
\label{eq:lift}
\end{equation}
with $u_n$ a known exogenous forcing coordinate, $\epsilon$ residual noise,
and $D\in\R^{h\times h}$ Schur stable (spectral radius $<1$) by construction.
Eliminating $z$ gives the exact convolution
$q_{n+1}=A q_n+\sum_{k\ge1}K_k\,q_{n-k}+\ldots$ with matrix-valued lag kernel
\begin{equation}
K_k=B\,D^{k-1}C,\qquad k\ge1 .
\label{eq:kernel}
\end{equation}
Local stability of \eqref{eq:lift} is governed by the block matrix
$M=\bigl(\begin{smallmatrix}A&B\\C&D\end{smallmatrix}\bigr)$.

\begin{proposition}[Schur-resolvent factorization and unit-multiplier criterion;
standard identity, stated for self-containedness]
\label{thm:schur}
For every $\zeta$ with $\zeta I-D$ invertible,
\begin{equation}
\det(\zeta I-M)=\det(\zeta I-D)\,\det\Phi(\zeta),
\qquad
\Phi(\zeta)=\zeta I-A-B(\zeta I-D)^{-1}C .
\label{eq:factor}
\end{equation}
Define the resolvent moments
$\mathcal{N}_j=B(I-D)^{-(j+1)}C\in\R^{d\times d}$. If $1\notin\mathrm{spec}(D)$,
then $M$ has a unit multiplier if and only if $\,\det\bigl(I-A-\mathcal{N}_0\bigr)=0$.
Moreover, when $\rho(D)<1$,
$\mathcal{N}_0=\sum_{k\ge1}K_k$,
$\mathcal{N}_1=\sum_{k\ge1}kK_k$, and
$\mathcal{N}_2=\sum_{k\ge1}\tfrac{k(k+1)}{2}K_k$,
so $\mathcal{N}_0$ is the integrated memory feedback and $\mathcal{N}_1,
\mathcal{N}_2$ encode its duration and curvature.
\end{proposition}

\begin{proof}
Equation \eqref{eq:factor} is the Schur determinant identity for the block
$\zeta I-M$ with invertible block $\zeta I-D$. At $\zeta=1$,
$\det(I-D)\neq0$, so the left side vanishes iff $\det\Phi(1)=0$, and
$\Phi(1)=I-A-\mathcal{N}_0$. The moment identities follow from the Neumann
series $(I-D)^{-(j+1)}=\sum_{k\ge0}\binom{k+j}{j}D^{k}$ applied to
\eqref{eq:kernel}.
\end{proof}

Proposition~\ref{thm:schur} is the conceptual core: for a projected observable,
the stationary stability boundary is an exact balance between instantaneous
drift ($A$) and integrated delayed feedback ($\mathcal{N}_0$). Scalar lag-one
autocorrelation conflates these two contributions; the operator diagnostic
separates them.

\subsubsection*{Slow multiplier and its moment reconstruction}

The near-unit multiplier admits an exact second-order reconstruction from the
first three moments $(\mathcal{N}_0,\mathcal{N}_1,\mathcal{N}_2)$ via a
Lyapunov--Schmidt reduction; the statement and proof are given in
Supplementary Note~1.

In this paper the \emph{primary} diagnostic is the exact eigenvalue of $M$
closest to $1$, computed directly from the fitted blocks; the reconstruction
serves as an internal consistency check, valid in its
asymptotic regime (pooled median absolute error $3.7\times10^{-3}$ in our
fits, degrading only where $\delta$ is not small, as expected).

\subsection{The forced collapse: comparative early warning}
\label{sec:upbranch}

\subsubsection{Classical indicators never alarm (preregistered)}

Causal lag-one autocorrelation never approaches its minimum admissible
threshold ($0.519$): its maximum over 1{,}357 evaluation years is $0.48$, and
it declines to negative values over the final pre-tipping decades
(Fig.~\ref{fig:upalarms}a). Variance exceeds its threshold only transiently
($0.555$ vs.\ $0.451$) and yields no valid sustained alarm
(Fig.~\ref{fig:upalarms}b). This replicates, with a strictly one-sided pipeline, the finding of
\cite{vanwesten2024}: critical slowing down of the scalar index is absent
here, so any successful indicator must exploit different structure.

\subsubsection{Causal $\FovS$: a 29-year lead costing 45 retractions (amended rule)}

The sequential turning-point rule issues its final unretracted alarm at model year
1729 (lead 29 years), bracketed by the retrospective references (raw minimum
1726, spline minimum 1732), at the cost, under our sequential implementation, of 45
retractions over the thirteen-century decline (Fig.~\ref{fig:upalarms}c). A real-time user applying this sequential implementation would have faced 45
retracted turning-point declarations before the final one; many are
provisional-minimum updates rather than classical false alarms, and it is
precisely this operational cost that retrospective evaluation does not reveal
and that, to our knowledge, had not been quantified.

\subsubsection{Operator point estimate: real episodes, no certified alarm
(preregistered/amended)}

The point-estimate multiplier (control threshold $0.99949$; zero sustained
control false alarms) exhibits genuine supercritical episodes: the operator
fitted at origin 1676 has exact multiplier $1.00399$ (a real eigenvalue beyond
the unit circle), and with 10-year refits two supercritical refits occur, at
1671 ($1.00516$) and 1731 ($1.00242$) (Fig.~\ref{fig:upalarms}d). Under the
frozen and amended rules no alarm is certified---the episodes are isolated at
refit resolution---and the preregistered primary gate is \emph{not evaluable},
since neither the operator nor the classical indicators produce a certified
alarm. We record this verdict unchanged. Its proximate cause is identifiable:
hyperparameters are selected by one-step forecast loss, to which the stability
diagnostic is far more sensitive than the loss itself, so the
selected-configuration trajectory flickers around threshold while the
underlying ensemble does not.

\begin{figure}[t]
\centering
\includegraphics[width=\textwidth]{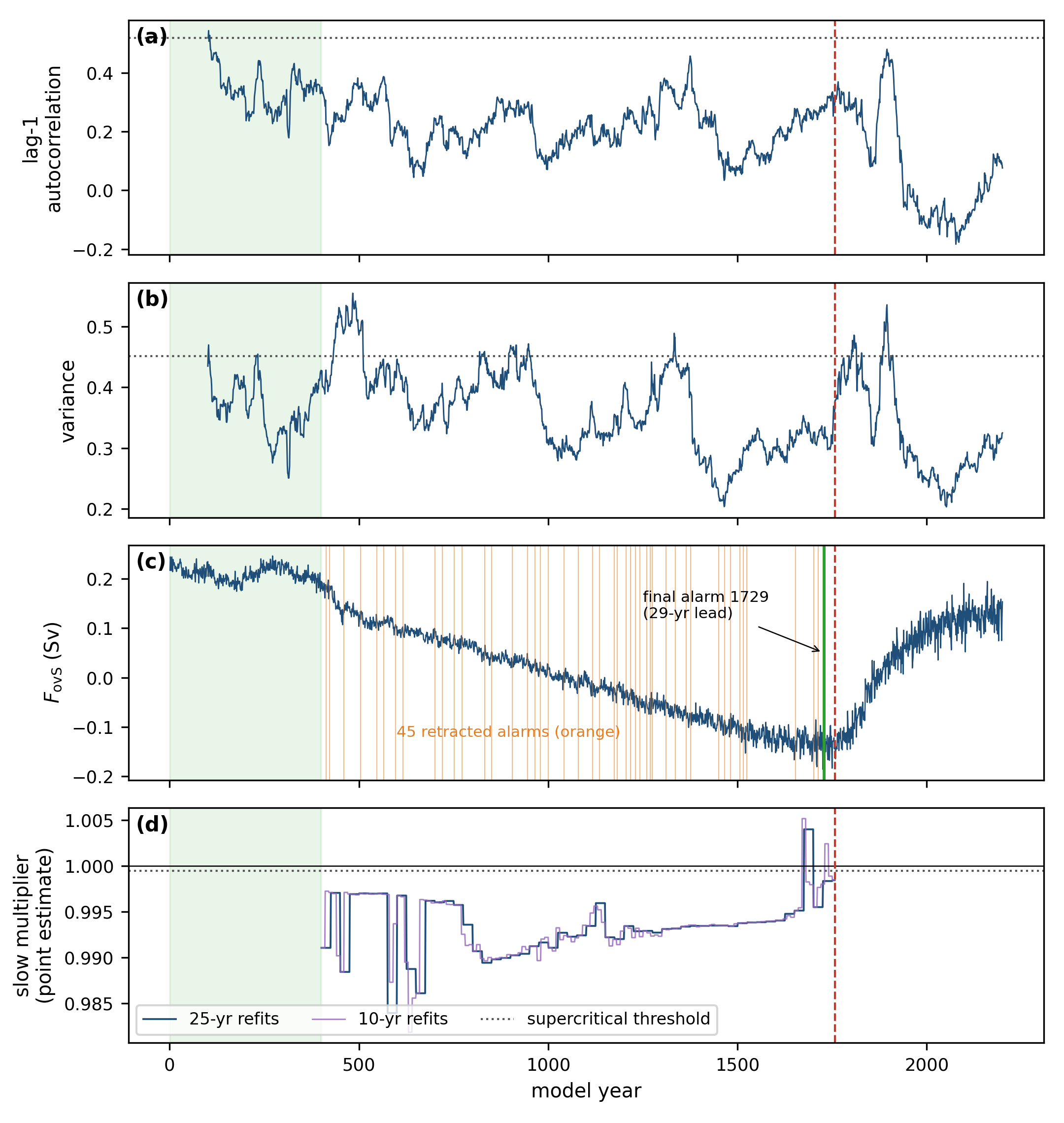}
\caption{Ascending branch, one-sided indicators against control-calibrated
thresholds (dotted); control shaded green, tipping year dashed red.
(a) lag-one autocorrelation; (b) variance; (c) $\FovS$ with the 45 retracted
turning points (orange) and the final alarm at 1729 (green); (d) operator
point-estimate multiplier for 25- and 10-year refits.}
\label{fig:upalarms}
\end{figure}

\subsubsection{The supercritical band is an ensemble property (exploratory)}

Evaluating all 24 configurations on fixed origins shows the pre-collapse
episode to be configuration-spanning: over 1650--1710 the supercritical
fraction is $0.42$--$0.75$, is $0.88$ at $\alpha=10^{-2}$ and $0.58$ at
$\alpha=10^{-1}$, with only the strongest ridge blind to it
(Fig.~\ref{fig:robustness}); on sparse control origins the same fixed-grid
fraction is $0$--$8\%$. The full ensemble fraction $\psi$ has a marked
temporal structure (Fig.~\ref{fig:psiup}): near zero for eight centuries
(801--1550), rising through 1551--1650, and sustained at mean $0.58$ through
the final century. The control is not silent---isolated decadal episodes of
internal variability (control years $\sim$230--290) reach instantaneous
$\psi$ up to $0.75$---so the \emph{level} of $\psi$ does not separate the
pre-collapse century from the strongest control excursions; \emph{duration}
does. The 50-year trailing (past-only) mean has control maximum $0.45$; the
evaluation series crosses the threshold $0.46$ at model year 1681 and never
returns below it before the collapse: an alarm with a 77-year lead and zero
sustained false alarms by construction. We label this alarm exploratory (the
duration statistic was identified after the frozen rules had been evaluated);
its blind validation follows.

\begin{figure}[t]
\centering
\includegraphics[width=.98\textwidth]{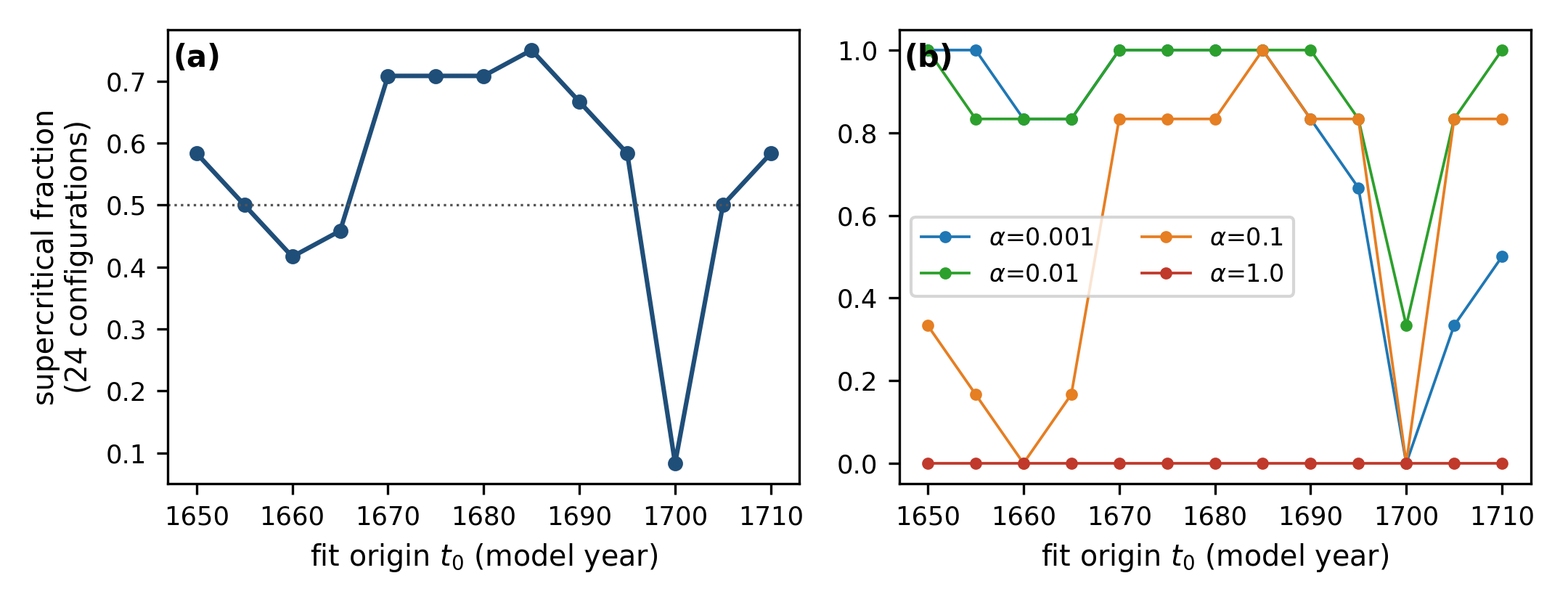}
\caption{Robustness of the pre-collapse supercritical band. (a) Fraction of
the 24 fixed configurations beyond the supercritical threshold, origins
1650--1710. (b) The same fraction stratified by ridge strength $\alpha$.}
\label{fig:robustness}
\end{figure}

\begin{figure}[t]
\centering
\includegraphics[width=\textwidth]{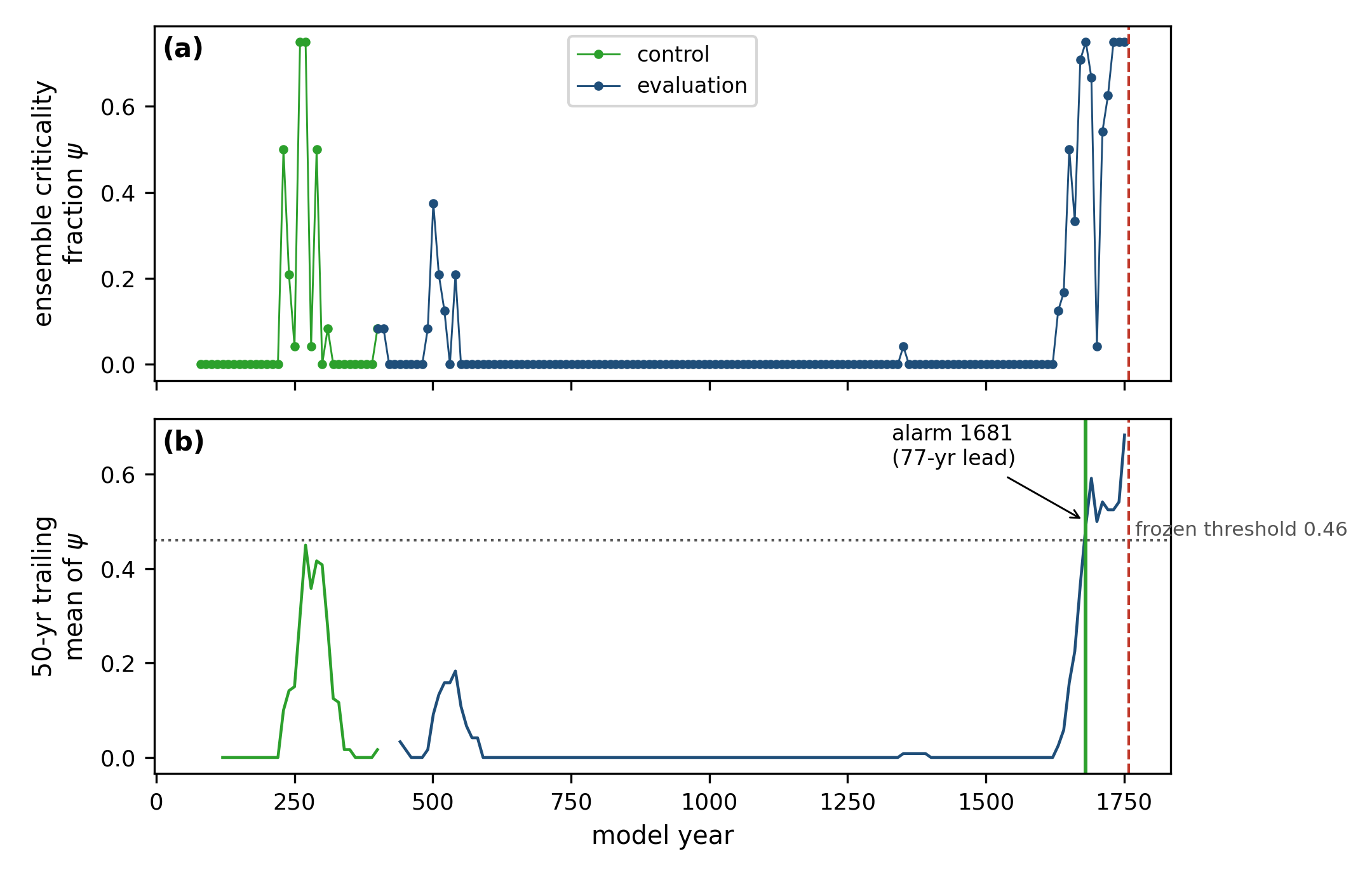}
\caption{(a) Ensemble criticality fraction $\psi$ for control (green) and
evaluation (blue). (b) Causal 50-year trailing mean; the frozen threshold
$0.46$ (dotted) is first exceeded at 1681, without retraction, 77 years before
the collapse (dashed red).}
\label{fig:psiup}
\end{figure}

\begin{table}[t]
\centering
\caption{Ascending-branch summary. ``Certified'' refers to frozen or amended
rules; the ensemble statistic is exploratory and is validated blindly in
Section~\ref{sec:validation}.}
\label{tab:summary}
\begin{tabular}{lccc}
\toprule
Indicator & Alarm year & Lead (yr) & Real-time cost \\
\midrule
Lag-1 autocorrelation (one-sided) & none & --- & --- \\
Variance (one-sided) & none & --- & --- \\
$\FovS$ sequential turning point & 1729 & 29 & 45 retracted alarms \\
Operator multiplier, point estimate & not certified & --- & 0 false alarms \\
Operator ensemble + duration (exploratory) & 1681 & 77 & 0 false alarms \\
\bottomrule
\end{tabular}
\end{table}

\subsection{Blind out-of-sample validation: the recovery transition}
\label{sec:validation}

\subsubsection{Frozen design and preregistered verdict}

Before any value of the descending branch was read (only file names were
listed), every component of the validation was frozen, committed, and publicly
timestamped (Amendment~2): the state contract, the 24-configuration ensemble,
the supercritical threshold transferred unchanged from the ascending control,
the trailing-mean threshold $0.46$, the certification rule, the origins
(2601--4391, every 10 years, minimum 400-year one-sided window), the ground-truth
procedure (two-segment break regression, computed last), and the success gate
(a certified alarm alive before the transition; zero certified-then-expired
alarms). The frozen run yields $t^{*}=4143$ and four certified alarms, all of
which expire: (2761--2811), (3191--3251), (3501--4131), and (4211--4321,
post-transition). No alarm is alive at $t^{*}$: \textbf{the validation gate
fails}, and the exploratory 77-year lead is not confirmed as a transferable
timing rule. We report this verdict unchanged.

\subsubsection{Labeled post-hoc diagnosis}
\label{sec:diagnosis}

The structure of the failure is informative (Fig.~\ref{fig:down}). The
collapsed state has index $1.4\pm1.7\Sv$ (years 2601--4050); its physical
departure precedes the frozen $t^{*}$, which dates the subsequent fast phase
(the recovery is roughly six times faster than the collapse
\cite{vanwesten2023}). Defining onset as the first sustained exceedance of
the collapsed-state mean plus three standard deviations gives onset year 4100
(index already $7$--$11\Sv$ there), and the identification is insensitive to
the convention: across fifteen variants ($2$--$4$ standard deviations,
persistence $3$--$10$ yr) the onset lies in 3984--4102
(Supplementary Table~1). Two facts then reframe the failure. First, the
certified alarm initiated at 3501 was alive at the physical onset under every
one of the fifteen conventions (trailing $\psi$ at onset $0.57$--$0.72$,
above threshold in $15/15$ cases); it expired at 4131, \emph{during} the
transit, precisely because $\psi$ collapsed ($1.00\to0.13\to0.00$ across
origins 4091--4121) once the system left the marginal collapsed state---the
expected behavior of a local stability meter during rapid transit. The formal
failure thus decomposes largely into a ground-truth convention effect. Second,
the six-century elevated episode (mean $\psi$ $0.94$ over 3501--3900, $0.74$
over 3901--4100) coincides with the collapsed state's slow approach to its
recovery saddle-node ($\FH\approx0.08$--$0.09\Sv$) and is therefore plausibly
a correct proximity reading rather than pure false alarm; against this, two
brief certified episodes (2761, 3191) are unambiguous false positives, and the
diagnostic is in alarm 45\% of the branch, showing that the transferred level
threshold is looser in the collapsed regime.

\begin{figure}[t]
\centering
\includegraphics[width=\textwidth]{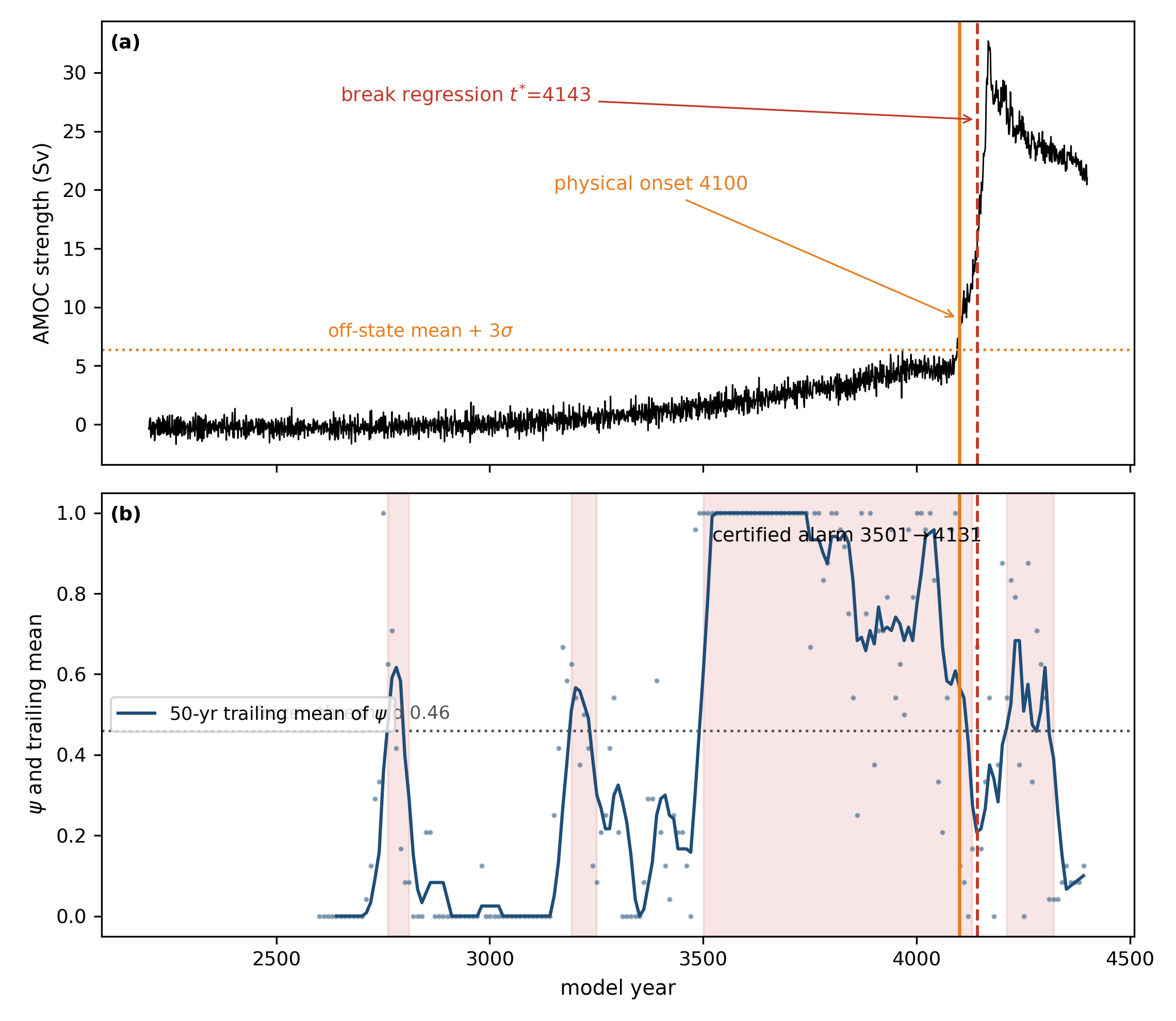}
\caption{Blind validation on the descending branch. (a) AMOC index with the
onset threshold (collapsed-state mean $+3\sigma$, orange dotted), the onset
year 4100 (orange), and the frozen break-regression $t^{*}=4143$ (red dashed).
(b) $\psi$ (dots), its 50-year trailing mean (line), the frozen threshold
(dotted), and certified alarm intervals (shaded). The 3501 alarm is alive at
the onset and expires during the rapid transit.}
\label{fig:down}
\end{figure}

\begin{remark}
All of Section~\ref{sec:diagnosis} is post-hoc and labeled as such; the onset
conventions were chosen after inspecting the descending branch, and
Supplementary Table~1 is a sensitivity analysis of the diagnosis, not a
passed gate. The preregistered verdict stands.
\end{remark}

\subsection{Preregistered follow-up: additional fixed-forcing and scenario
simulations}
\label{sec:protocolT}

The lessons of the two preceding sections were converted into a second frozen
experiment (Protocol T; Supplementary Note~2) on the additional
simulations of \cite{vanwesten2025}: a control branch far from the fold
(CESM\_0600), five fixed-forcing branches anchored at quasi-equilibrium years
1500--1700 (increasingly close to the fold; collapses, if any, are
noise-induced at constant forcing), and four transient scenario runs (RCP4.5
and RCP8.5 on two base states; finite forcing speed). Only scalar AMOC series
are archived for these runs, so the protocol prespecified the $d=1$ variant of
the operator, thresholds recalibrated on the same-experiment control,
onset-based ground truth as primary, a proximity-discrimination gate (T2) for
the fixed-forcing branches, and a per-run timing gate (T3) for the scenario
runs. Run lengths and outcomes were unknown at freezing; the protocol was
committed and publicly timestamped before any data value was read.

\textbf{Timing at finite forcing speed (T3): four for four.} In every
transient scenario run the frozen statistic certified an alarm before the
collapse onset: leads of 34 and 35 years for the RCP4.5 and RCP8.5 scenarios
from the CESM\_0600 base state (onsets 2029 and 2030), and 10 and 8 years for
the corresponding scenarios from the CESM\_1500 base state (onsets 2055 and
2043); in all four the trailing statistic remained above threshold
continuously from alarm to onset, rising to ensemble fractions of up to
$1.0$ (Fig.~\ref{fig:protocolT}a; Table~\ref{tab:protocolT}). Notably, the
two CESM\_0600 scenarios share their pre-divergence segment, and the alarm
(year 1995) fired within it: a single alarm was thereby validated against two
independently evolving outcomes. Against these detections stand zero
certified alarms in the three simulations without a transition---the control
branch and the two non-collapsing fixed-forcing branches, roughly 1{,}400
run-years in total; the largest excursion in any of them (a single-origin
peak in CESM\_1550) did not certify.

\textbf{Proximity discrimination for rapid noise-induced collapses (T2):
failed, and reported unchanged.} The three collapsing fixed-forcing branches
transition within roughly 80--130 years of branching, leaving 0, 8, and 17
evaluable pre-onset origins; one is window-limited outright, and the branch
with the longest usable record (CESM\_1600, 85 years of pre-onset trailing
values) shows a genuine null: its trailing statistic is zero throughout. Gate
T2 therefore fails. Subsequent discriminating tests (Section~\ref{sec:ablations}) show the null is
not explained by record length---prepending the full quasi-equilibrium history
leaves the statistic at zero---and, via an observability ablation on the
ascending branch, that reducing the observable to the scalar index loses the
alarm under identical rules. Limited observable richness is thereby identified as a major determinant of
detectability in these records.

\begin{table}[t]
\centering
\caption{Protocol T outcomes (all rules and thresholds frozen before
execution; outcomes unknown at freezing). Onset: first sustained 5-year
excursion below the run-reference mean $-3\sigma$.}
\label{tab:protocolT}
\begin{tabular}{lcccc}
\toprule
Run & Outcome & Onset & Certified alarm & Lead (yr) \\
\midrule
Control CESM\_0600 & no transition & --- & none & --- \\
Branch 1500 (fixed $\FH$) & no transition & --- & none & --- \\
Branch 1550 (fixed $\FH$) & no transition & --- & none & --- \\
Branch 1600 (fixed $\FH$) & collapse & 1803 & none & --- \\
Branch 1650 (fixed $\FH$) & collapse & 1782 & none & --- \\
Branch 1700 (fixed $\FH$) & collapse & 1781 & (window-limited) & --- \\
Scenario 0600 RCP4.5 & collapse & 2029 & 1995 & 34 \\
Scenario 0600 RCP8.5 & collapse & 2030 & 1995 & 35 \\
Scenario 1500 RCP4.5 & collapse & 2055 & 2045 & 10 \\
Scenario 1500 RCP8.5 & collapse & 2043 & 2035 & 8 \\
\bottomrule
\end{tabular}
\end{table}

\begin{figure}[t]
\centering
\includegraphics[width=.98\textwidth]{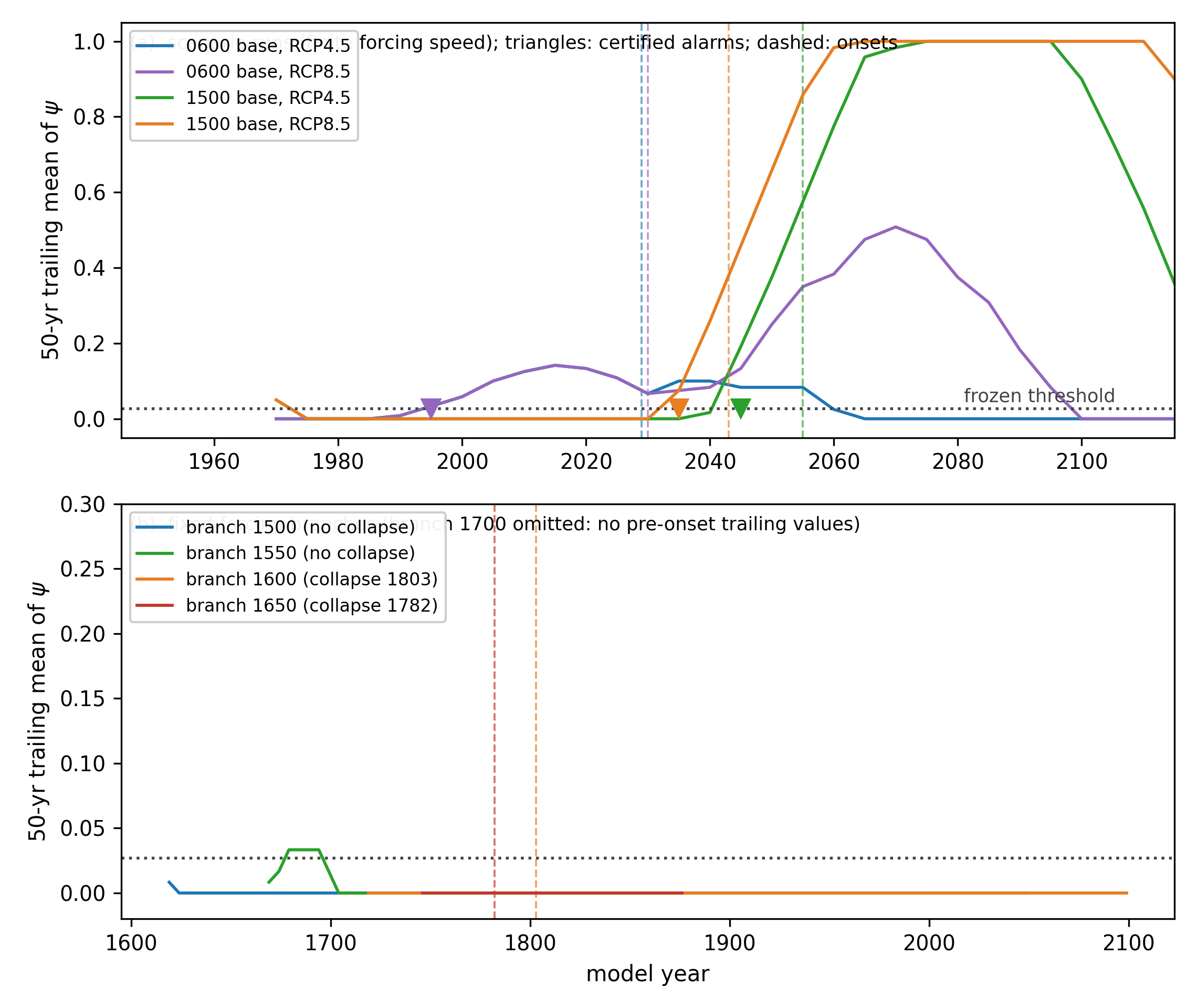}
\caption{Protocol T. (a) Trailing ensemble statistic for the four transient
scenario runs; triangles mark certified alarms, dashed lines the collapse
onsets, dotted line the frozen threshold. (b) The same statistic for the
fixed-forcing branches: near zero throughout, including the collapsing branch
1600 (branch 1700 is omitted: it collapses before any pre-onset trailing value
is defined).}
\label{fig:protocolT}
\end{figure}

\subsection{Ablations, matched comparators, and estimation uncertainty}
\label{sec:ablations}

Six diagnostic analyses, run after the frozen experiments and labeled as such,
address the identification of the signal
(Fig.~\ref{fig:ident}a; Table~\ref{tab:ablations}). First,
\emph{trend and forcing ablations} on the ascending branch, under identical
windows, origins, thresholds recalibrated per variant, and rules: removing a
within-window linear trend from every state variable and profile level leaves the
alarm essentially unchanged (1671 vs.\ 1681; zero control false alarms), while
removing or misaligning the forcing covariate makes the indicator alarm far
\emph{earlier} (961 and 1131)---the covariate protects the diagnostic from
drift contamination rather than creating the signal. Second, the
\emph{matched-information comparators}: setting $B=0$ in the same model (same
eight-dimensional state, forcing, ridge, windows, and rules---a first-order vector
autoregression with exogenous forcing) never alarms, its slow multiplier remaining
at $0.91$--$0.93$ through the final pre-collapse century, and an explicit-lag
VAR(3) with the same state and forcing is likewise blind (multipliers
$0.967$--$0.974$, no alarm); the added value is specifically the long-timescale
kernel structure, not multivariate information or lag order per se. Third, \emph{estimation validation} in
systems with known Jacobians (Fig.~\ref{fig:estval}): in a ramped Stommel model
the estimator preserves part of the ordering and the approach-to-criticality
trend (pooled correlation $0.86$ across three realizations) but exhibits severe
upward bias and range compression (calibration slope $0.13$, intercept $0.89$,
MAE $0.34$); control-relative thresholding substantially reduces the operational
effect of this bias in the controlled experiments examined here---the
pipeline-faithful version, calibrated on a stationary control segment, alarms in
orderly fashion as true criticality is approached. Tracking degrades with ramp
speed (correlation $0.89$, $0.86$, $0.68$ for 3{,}000-, 1{,}500-, and 750-year
ramps), the mirror image of the quasi-static limit: too slow and lead time is
ill posed, too fast and estimation quality falls. Under partial observation
($d=1$, hidden fast modes) the estimated multiplier is additionally attenuated (a
true change of $0.18$ compresses to $0.02$), quantifying a real limitation of
scalar records. Fourth, \emph{uncertainty}: moving-block residual bootstrap (200
replicates, block length 10, recursive regeneration under the fitted closed loop)
places the principal point-estimate episode entirely above unity, as reported in
the point-estimate subsection; we emphasize its scope---it certifies
$\hat\rho_{\mathrm{slow}}(1676)>1$, not the 1681 ensemble alarm. The alarm
year itself carries a separate, resampling-based uncertainty: block bootstrap of
the control fraction (threshold uncertainty) combined with block resampling of
the evaluation-series residuals (500 draws) yields an alarm in $98\%$ of
replicates, with median year 1691, interquartile range 1681--1701, and $90\%$
interval 1661--1711---a lead of 47--97 years. Composition sensitivity is
addressed next.

Fifth, an \emph{observability ablation on the ascending branch itself}, where the
full profiles exist, tests whether observable richness controls detectability
under identical rules (Fig.~\ref{fig:ident}b): the full $d=8$ state alarms at
1681; reducing to the scalar index ($d=1$) loses the alarm entirely; and the
intermediate reductions show the dependence is on \emph{which} coordinates rather
than on dimension count---the five physical coordinates without profile EOFs do
not certify (their control is noisier), while a two-coordinate state containing
the deep-return transport crosses early and sustains. The signal resides in the
vertical overturning structure, which the scalar index does not transmit. Sixth,
\emph{ensemble composition}: removing any regularization family, any timescale
family, or the oscillatory unit moves the alarm by at most ten years
(1671--1681), equal weighting by regularization family reproduces 1681 exactly,
and a continuous median-margin statistic yields a later, more conservative alarm
(1741) with zero control alarms; the 1681 alarm is a property of the ensemble,
not of its count.

Two further checks sharpen the Protocol T interpretation. Prepending the full
quasi-equilibrium history to the fixed-forcing branches (1{,}600--1{,}700 years
of additional record) leaves their pre-onset statistic at zero: the null is not
explained by record length, and together with the observability ablation above it
is best explained by the scalar observable, which does not transmit the vertical
structure in which the signal resides. And the matched scalar VARX applied to the
four scenario runs splits the finite-speed evidence in an instructive way: from
the strong base state a memoryless loss-of-contraction signal suffices and indeed
leads \emph{earlier} (79--80 vs.\ 34--35 years), while from the weakened base
state the memoryless detector is blind ($\psi_{\mathrm{VARX}}$ maximum $0.40$
and $0.00$) where the memory operator certified 10- and 8-year leads. The
evidence is therefore not that the memory operator dominates memoryless detectors
globally; it is that it matches them where a memoryless signal is present and
detects where it is absent. With three transition-free runs and four detected
transitions, exact binomial limits remain wide (two-sided 95\%: false-alarm
probability up to $0.71$, sensitivity from $0.40$); these counts establish
consistency, not calibrated rates.

\begin{figure}[t]
\centering
\includegraphics[width=.98\textwidth]{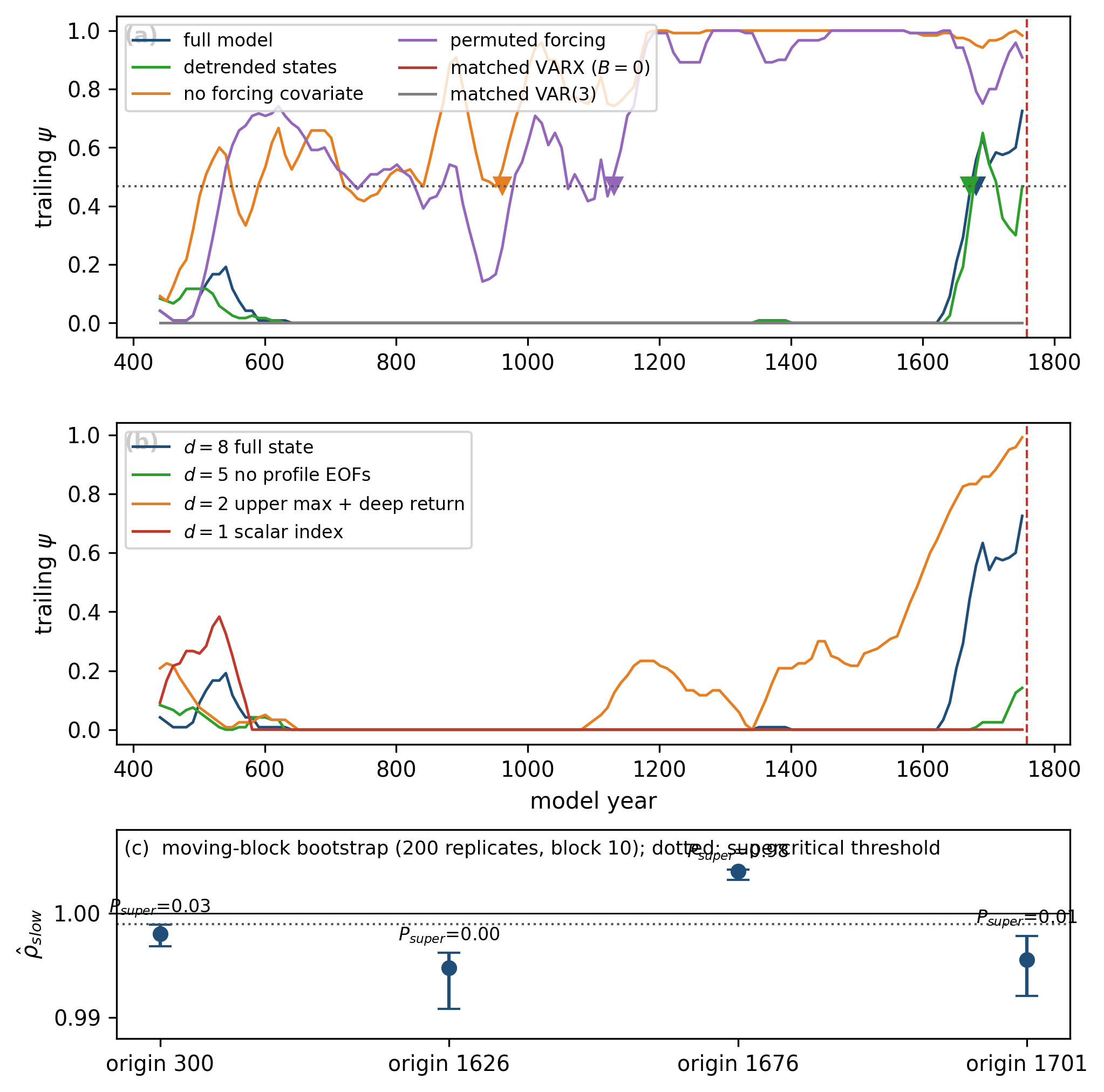}
\caption{Identification of the signal (ascending branch; identical origins,
windows, and rules; thresholds recalibrated per variant). (a) Trailing ensemble
statistic for the full model, detrended states, removed and permuted forcing,
and the matched memoryless comparators (VARX, VAR(3)); triangles mark certified
alarms. (b) Observability ablation: the alarm is lost when the observable is
reduced to the scalar index, and the dependence is on which coordinates carry
the vertical overturning structure rather than on dimension count. (c)
Moving-block bootstrap of the point-estimate multiplier at a control origin and
at the episode and its neighbors: the 1676 interval lies entirely above unity.}
\label{fig:ident}
\end{figure}

\begin{table}[t]
\centering
\caption{Ablations and matched comparators (diagnostic, post-freeze; identical
windows and rules throughout; thresholds recalibrated per variant on the same
control).}
\label{tab:ablations}
\begin{tabular}{lccc}
\toprule
Variant (ascending branch, $d=8$) & Alarm year & Lead (yr) & Sustained control alarms \\
\midrule
Full model (baseline) & 1681 & 77 & 0 \\
Detrended states & 1671 & 87 & 0 \\
No forcing covariate & 961 & --- & 0 \\
Permuted forcing covariate & 1131 & --- & 0 \\
Matched VARX ($B=0$) & never & --- & 0 \\
Matched VAR(3), explicit lags & never & --- & 0 \\
\midrule
Scenario runs ($d=1$) & MZ lead & VARX lead & \\
\midrule
0600 RCP4.5 / RCP8.5 & 34 / 35 & 79 / 80 & \\
1500 RCP4.5 / RCP8.5 & 10 / 8 & none / none & \\
\bottomrule
\end{tabular}
\end{table}

\begin{figure}[t]
\centering
\includegraphics[width=.95\textwidth]{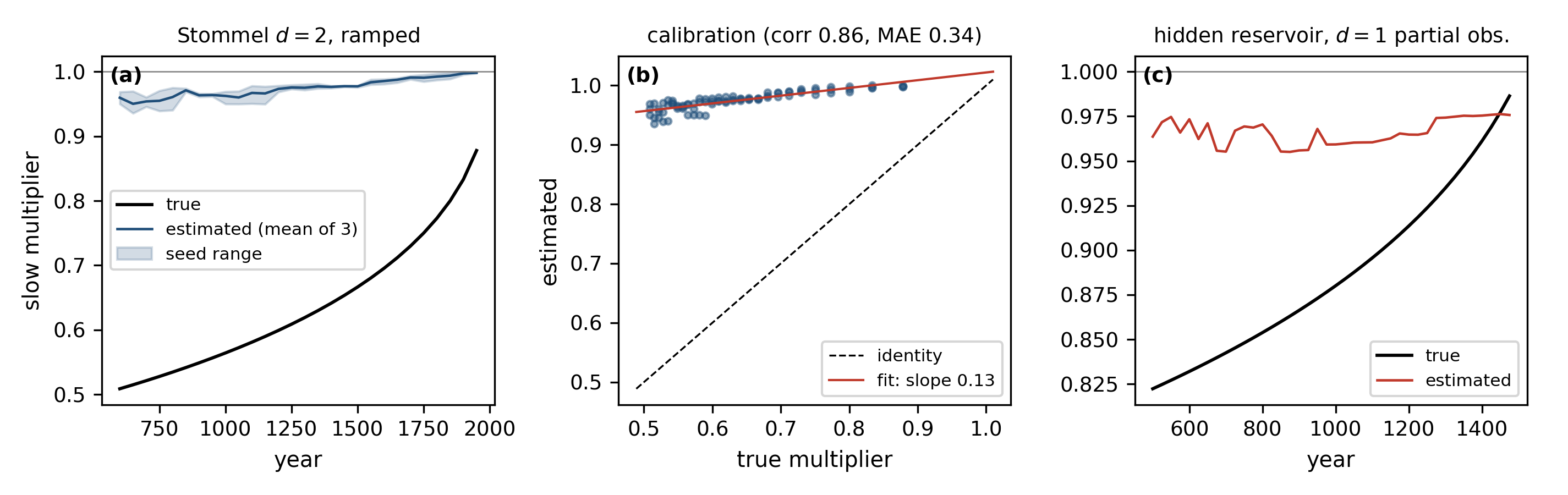}
\caption{Estimation validation in systems with known Jacobians. Left: ramped
Stommel model (full state, $d=2$): the estimated slow multiplier tracks the true
one with a bias toward unity that control-relative thresholding absorbs. Right:
hidden-reservoir model under partial observation ($d=1$): strong attenuation of
the estimated multiplier, quantifying the scalar-record limitation.}
\label{fig:estval}
\end{figure}

\section{Discussion}
\label{sec:discussion}

Across both transitions the diagnostic behaves consistently as a meter of
proximity to local criticality of the fitted effective dynamics: silent for
centuries of strong restoring, robustly elevated on the approach to a fold
(the pre-collapse century; the collapsed state's final six centuries), and
extinguished once the system is in rapid transit. What it is not, as
certified here, is a clock. We state this as a delimitation, not a
consolation: it fixes which claims the present evidence supports (graded,
mechanistically decomposable proximity) and which it does not (calibrated
lead times transferable across regimes).

Three separable causes emerged. (i) \emph{Selection noise}: hyperparameters
chosen by one-step loss make the point estimate flicker around threshold
while the ensemble is stable; certification applied to the selected
configuration inherits the flicker. (ii) \emph{Threshold transfer}: a
supercritical threshold calibrated in the on-state control is demonstrably
looser in the collapsed regime. (iii) \emph{Ill-posedness of lead under
quasi-static approach}: when forcing crawls toward a saddle-node the system
is genuinely near-critical for centuries and transition timing is set by
noise; a lead time is well defined only for transitions crossed at finite
forcing speed, which is precisely the ascending-branch regime in which the
77-year exploratory lead was obtained. A fourth, established by the ablation
round, is \emph{observable richness}: under partial (scalar) observation the
estimated multiplier is strongly attenuated, and detections there require
large true changes.

The follow-up experiment of Section~\ref{sec:protocolT} closes the loop
opened by these three causes: at finite forcing speed---the regime our analysis identified as the one where
lead time is well posed---the frozen statistic certified alarms in four scenario
trajectories spanning two initial-state families,
whereas rapid noise-induced escapes from short near-critical records
(the fixed-forcing branches) yielded no precursor, consistent with the
quasi-static lesson of the descending branch. The operating envelope of the indicator is thereby empirically supported, and
delimited on both sides, within the simulations examined here.

On the cleanest available simulated collapse, the comparative landscape is
now: classical scalar indicators fail for structural reasons, replicated with a
one-sided pipeline here; the physics-based $\FovS$ works with a quarter-century lead
but carries a newly quantified real-time cost of 45 retractions under our
sequential implementation and presupposes its mechanism; the memory-operator ensemble registers the
approach to criticality earlier than either, without prescribing a specific
feedback mechanism, with an
interpretation---delayed feedback eroding effective restoring---that
complements the salt-advection diagnostic. A pragmatic reading is that
proximity meters of the present type and physics-based indicators belong
together in a monitoring dashboard, the former for graded early warning, the
latter for mechanism-specific confirmation.

The pre-industrial control itself contains decadal episodes (years
$\sim$230--290) that the operator reads as near-critical; whether these
correspond to known modes of multidecadal AMOC variability in CESM is left
open. Everything here concerns one model and simulated transitions; nothing
licenses statements about the real ocean's tipping schedule. In particular, the
indicator measures proximity to transitions of the \emph{fitted effective
dynamics}---in the CESM experiments, transitions into strongly weakened
states---and is agnostic about whether basin-scale stabilizing pathways preclude a
complete collapse of the real circulation \cite{baker2025}; elevated proximity
readings and a dynamical floor of the kind reported by \cite{baker2025} can
coexist. The contribution
is methodological: a fully preregistered, blindly validated,
failure-inclusive assessment of a new indicator class on ground-truthed
transitions, reproducible from public data and timestamped code.

In sum, three messages emerge. First, delayed feedback carries stability
information that memoryless models with identical inputs do not recover in the
regimes where classical indicators fail: the matched first-order and
third-order autoregressions remain far from criticality throughout the
pre-collapse century, and are blind in the weakened-base scenarios where the
memory operator still certifies alarms. Second, the signal resides in the
vertical overturning structure and attenuates strongly under scalar
observation---connecting our observability ablation to the emerging question
of which resolved coordinates carry the critical mode
\cite{lohmann2025,smolders2025}. Third, whether proximity can be converted
into calibrated timing beyond the finite-rate simulations examined here is a
frozen, registered, and falsifiable question: generalization requires
initial-condition ensembles and additional models, and every rule needed to
test it is already public.

\section{Methods}
\label{sec:methods}

\subsection{Data and parity verification}
\label{sec:data}

All data are annual means from the 4{,}400-year CESM quasi-equilibrium
hysteresis experiment \cite{vanwesten2023,vanwesten2024}, from the authors'
public archives (Zenodo \texttt{10.5281/zenodo.10461549},
\texttt{10.5281/zenodo.8262424}; CC-BY 4.0). The North Atlantic freshwater
forcing $\FH$ increases linearly at $3\times10^{-4}\Sv\,\mathrm{yr}^{-1}$ from
$0$ to $0.66\Sv$ over model years 1--2200 (ascending branch; collapse at model
year 1758 by the published break-regression convention) and decreases
symmetrically over years 2201--4400 (descending branch), with recovery at much
weaker forcing, $\FH\approx0.08\Sv$, reflecting the published hysteresis
asymmetry \cite{vanwesten2023}. Figure~\ref{fig:overview} shows the full
experiment.

We use the archived overturning streamfunction
$\Psi(\mathrm{year},\mathrm{depth},\mathrm{lat})$ (61 levels, 310 latitudes),
the raw meridional-velocity section at $26^\circ$N, and the $34^\circ$S
velocity--salinity section from which $\FovS$ is computed. The scalar AMOC
index is the upper-1000-m transport at $26^\circ$N, recomputed by an
independent reimplementation of the authors' processing (partial bottom cells
included): agreement with the archived ascending-branch series is
$1.8\times10^{-15}\Sv$. On the descending branch the residual is
$2.0\times10^{-7}\Sv$; we verified that the two official archives themselves
disagree by up to $5.5\times10^{-8}\Sv$ on their 2{,}200 common years, so this
residual reflects the archives' numerical provenance, not the loader
(Amendment~3, Supplementary Note~2).

\subsection{Operator realization, fitting, and the ensemble statistic}
\label{sec:realization}

The memory bank comprises three exponential units with fixed timescales
$\tau$ and, optionally, one damped rotation (period 40 yr, damping $0.9$),
each carrying the full vector state; $D$ is block diagonal and Schur stable by
construction, $C$ the corresponding drive, and $(A-I,B,F,b_0)$ are fitted by
ridge regression on one-step increments within a strictly one-sided window
(standardization and EOFs likewise one-sided). Hyperparameters range over a fixed
grid of 24 configurations: timescale sets
$\{(1,5,20),(1,6,28),(2,8,30)\}$ yr, rotation on/off, ridge
$\alpha\in\{10^{-3},10^{-2},10^{-1},1\}$. Two diagnostics follow: the
\emph{point estimate} (multiplier of the configuration selected by one-sided
inner validation at each refit) and the \emph{ensemble criticality fraction}
$\psi(t_0)\in[0,1]$, the fraction of the 24 configurations whose exact
multiplier meets a fixed supercritical threshold; the duration statistic is
the 50-year trailing (past-only) mean of $\psi$.

The resolved state is an eight-dimensional description of the overturning
geometry at $25.80^\circ$N: upper-cell maximum, depth of the maximum,
deep-return minimum (1500--4500 m), mean upper (0--1000 m) and deep
(1500--4500 m) transports, plus the leading three one-sided EOF amplitudes of the
500--4500 m profile. The exogenous coordinate is the known ramp $\FH(t)$.

\subsection{Comparators, alarm rule, and preregistration}
\label{sec:rules}

\emph{Classical indicators}: lag-one autocorrelation and variance of the
scalar index in a trailing (past-only) 100-year window with within-window linear
detrending. \emph{Physics-based}: the archived $\FovS$ with a one-sided
turning-point rule mirroring the published retrospective minimum---alarm when
the running minimum is unprecedented relative to the 400-year control and the
series has spent three consecutive years at least one control standard
deviation above it; \emph{retracted} whenever a later new minimum occurs. We
report the final unretracted alarm and the retraction count as the
indicator's real-time cost.

All thresholds are calibrated only on the control segment (years 1--400,
$\FH\le0.12\Sv$) at a sustained-false-alarm budget of 5\% per control century.
Certification for block-wise diagnostics operates at refit resolution
(two consecutive refits at threshold; at most one sub-threshold refit
tolerated afterwards; Amendment~1). Evaluation origins are annual over
401--1757 with refits every 25 years (10 in sensitivity); lead is measured to
the published tipping year 1758. The protocol, amendments, code, and results
are committed with annotated tags whose chronology certifies design
$\to$ freeze $\to$ computation $\to$ results; each result below is labeled
\emph{preregistered}, \emph{amended}, or \emph{exploratory}
(Supplementary Note~2).

\subsection{Post-freeze ablations and matched comparators}
\label{sec:m_ablations}

All ablations use the origins, windows, certification rules, and false-alarm
budget of the main pipeline; only the fit input changes, and the supercritical
threshold is recalibrated per variant on the control segment by the frozen
procedure. \emph{Detrended}: within each one-sided window, every state variable and
every profile level is replaced by its residual from an ordinary least-squares
regression on an affine function of time computed over that window only, before
standardization and EOFs. \emph{No forcing}: the exogenous covariate is
identically zero. \emph{Permuted forcing}: the covariate is the $\FH$ series
circularly shifted by $+500$ years (a single deterministic misalignment null).
\emph{Matched VARX}: the same one-step regression with the memory columns
removed ($B\equiv0$), inner selection over the ridge grid by one-step validation
error; its diagnostic is the eigenvalue of $A$ closest to 1. \emph{Matched
VAR(3)}: three explicit state lags with the same forcing and intercept; the
diagnostic is the companion-matrix eigenvalue closest to 1. The observability
ablation refits the full pipeline on nested subsets of the state: the five
physical coordinates without EOFs ($d=5$); upper-cell maximum plus deep-return
transport ($d=2$); and the scalar index ($d=1$).

\subsection{Controlled validation in systems with known stability}
\label{sec:m_controlled}

\emph{Stommel model}: $\dot T=\eta_1-T-|T-S|T$, $\dot
S=\eta_2-\eta_3S-|T-S|S$ with $\eta_1=3$, $\eta_3=0.3$;
Euler--Maruyama with $\Delta t=0.05$, additive noise of standard deviation
$\sigma$ on both variables, annual sampling; forcing protocol: 500 stationary
years at $\eta_2=1.0$ (control) followed by a linear ramp to $\eta_2=1.2199$
over 1{,}500 years (the fold is at $\eta_2\approx1.2201$); three noise
realizations at $\sigma=0.02$ for calibration statistics, and single
realizations of 750- and 3{,}000-year ramps for the speed dependence. \emph{Hidden
reservoir}: $\dot y=r-y^2-b_1s_1-b_2s_2+\sigma_y\dot W$, $\dot
s_i=c_i(\sqrt r-y)-\epsilon_i s_i$ with
$(b_1,b_2,c_1,c_2)=(0.18,0.08,0.25,0.12)$,
$\epsilon_1=0.04+0.45\sqrt r$, $\epsilon_2=0.65$; only $y$ observed; 400
stationary years at $r=0.25$ then a ramp to $r=0.037$ (fold at
$r_c\approx0.0353$); trajectories are truncated if $y<-0.2$ (basin escape).
The true multiplier is $\rho_{\mathrm{true}}(t)=\exp[\lambda_{\max}(J(t))\,
\Delta]$ with $J(t)$ the Jacobian at the equilibrium of the frozen instantaneous
forcing and $\Delta=1$ year. The pipeline-faithful evaluation calibrates the
per-configuration threshold and the trailing threshold on the stationary control
segment exactly as in the main experiment.

\subsection{Moving-block bootstrap}
\label{sec:m_bootstrap}

At a given origin, the selected configuration's residuals are resampled in
moving blocks of length 10; a surrogate trajectory is regenerated recursively
under the fitted closed-loop map driven by the resampled residuals; the same
configuration is refitted to the surrogate; 200 replicates per origin. Reported
origins: one control (300) and the episode and its neighbors (1626, 1676, 1701).
Hyperparameter selection, standardization, and EOFs are held fixed within
replicates, so the intervals quantify residual-driven estimation noise for the
selected configuration; The alarm-year distribution is obtained separately, without refits: circular
block bootstrap (block of 5 origins) of the control fraction series regenerates
the trailing threshold, block resampling of the evaluation-series residuals
about its trailing mean regenerates the evaluation statistic, and the frozen
certification rule is reapplied; 500 draws. A full refit-level bootstrap of the
ensemble fraction remains future work.

\subsection{Protocol T scalar comparators}
\label{sec:m_protocolT_cmp}

The scalar VARX for the scenario runs uses the standardized index, the nominal
scenario ramp as covariate, inner selection over the ridge grid, the
control-branch calibration of Protocol T, and the same origins and rules as the
memory variant; the quasi-equilibrium prepend test concatenates the archived
ascending-branch index (years 1 to the branch anchor) with each branch series
and repeats the $d=1$ pipeline with thresholds recalibrated on the
correspondingly prepended control branch.

\backmatter

\bmhead{Supplementary information}
Supplementary Notes 1--2 and Supplementary Table 1 are provided as a single
separate Supplementary Information file: Supplementary Note 1 (statement and proof of the slow-multiplier expansion theorem), Supplementary Table 1 (robustness of the recovery
onset convention), Supplementary Note 2 (preregistration ledger).

\bmhead{Data availability}
All input data are the public archives of van Westen and coworkers (Zenodo
\texttt{10.5281/zenodo.10461549} and \texttt{10.5281/zenodo.8262424}; CC-BY
4.0). All derived data, results files, and figures of this manuscript are
permanently archived at Zenodo (DOI \texttt{10.5281/zenodo.21446597}).

\bmhead{Code availability}
The complete analysis code---ingestion with parity evidence, operator and
diagnostic modules, frozen protocols and amendments with their public
timestamps, and the ablation suite---is available at
\texttt{github.com/mauricio-herrera/EWS-CESM-MZ} and archived at Zenodo (DOI
\texttt{10.5281/zenodo.21446597}).

\bmhead{Acknowledgements}
Not applicable.

\bmhead{Funding}
The author received no specific funding for this work.

\bmhead{Author contributions}
M.H.-M.\ conceived the study, developed the theory, performed all analyses,
and wrote the manuscript.

\bmhead{Competing interests}
The author declares no competing interests.

\end{document}